\documentclass[12pt]{amsart}
\usepackage{graphicx} 
\usepackage{setspace} 
\usepackage{amsfonts,amssymb,latexsym,amsmath,amsthm}
\usepackage{enumerate} 
\usepackage{ulem} 
\usepackage{currfile} 
\usepackage{hyperref} 
\usepackage[margin=1in]{geometry}
\usepackage{dsfont} 
\usepackage{bbm} 
\usepackage{enumerate} 
\newtheorem{theorem}{Theorem}[section]
\newtheorem{definition}[theorem]{Definition}

\newtheorem{lemma}[theorem]{Lemma}
\newtheorem{remark}[theorem]{Remark}

\newtheorem{example}[theorem]{Example}
\numberwithin{equation}{section} 
\usepackage{mathtools}

\usepackage[capitalize,nameinlink,noabbrev]{cleveref}
\usepackage{tikz,pgfplots}
\usetikzlibrary{pgfplots.statistics}
\pgfplotsset{compat=newest}

\crefname{ineq}{Inequality}{Inequalities}
\creflabelformat{ineq}{#2{\upshape(#1)}#3}

\DeclarePairedDelimiter\innprod{\langle}{\rangle}
\newcommand{\dotprod}[2]{\innprod*{#1,#2}}

\DeclarePairedDelimiter\abs{\lvert}{\rvert}
\DeclarePairedDelimiter\set{\{}{\}}
\DeclarePairedDelimiter\ceil{\lceil}{\rceil}

\DeclarePairedDelimiter\brackets{[}{]}
\DeclarePairedDelimiter\norm{\|}{\|}

\renewcommand{\vec}[1]{#1}
\newcommand{\NN}{\mathbb{N}}

\newcommand{\RR}{\mathbb{R}}
\newcommand{\CC}{\mathbb{C}}

\newcommand{\veccoll}[4]{\{ #1_{#2} \}_{#3}^{#4}}
\newcommand{\tr}{\operatorname{tr}}

\allowdisplaybreaks 

\begin{document}


\title{Low coherence unit norm tight frames}
\author{Somantika Datta}
\address{Department of Mathematics, University of Idaho, Moscow, ID 83844-1103, USA}
\email{sdatta@uidaho.edu}
\author{Jesse Oldroyd}
\address{Department of Mathematics and Computer Science, West Virginia Wesleyan College, Buckhannon, WV 26201, USA}
\email{joldroyd.j@wvwc.edu}


\maketitle

\begin{abstract}
Equiangular tight frames (ETFs) have found significant applications in signal processing and coding theory due to their robustness to noise and transmission losses. ETFs are characterized by the fact that the coherence between any two distinct vectors is equal to the Welch bound. This guarantees that the maximum coherence between pairs of vectors is minimized. Despite their usefulness and widespread applications, ETFs of a given size $N$ are only guaranteed to exist in $\mathbb{R}^{d}$ or $\mathbb{C}^{d}$ if $N = d + 1$. This leads to the problem of finding approximations of ETFs of $N$ vectors in $\mathbb{R}^{d}$ or $\mathbb{C}^{d}$ where $N > d+1.$ To be more precise, one wishes to construct a unit norm tight frame (UNTF) such that the maximum coherence between distinct vectors of this frame is as close to the Welch bound as possible. In this paper low coherence UNTFs in $\mathbb{R}^d$ are constructed by adding a strategically chosen set of vectors called an \textit{optimal} set to an existing ETF of $d+1$ vectors. In order to do so,  combinatorial objects called block designs are used. Estimates are provided on the maximum coherence between distinct vectors of this low coherence UNTF. It is shown that for certain block designs, the constructed UNTF attains the smallest possible maximum coherence between pairs of vectors among all UNTFs containing the starting ETF of $d+1$ vectors. This is particularly desirable if there does not exist a set of the same size for which the Welch bound is attained.
\end{abstract}


\section{Introduction} \label{Intro}
\subsection{Background and Motivation}\label{subsection :: background and motivation}

The maximum coherence between pairs in a set of $N$ unit vectors $\{f_1, \ldots, f_N\}$ in $\mathbb{C}^d$ satisfies the following inequality due to Welch \cite{W1}:
\begin{equation} \label{WelchInEq}
\max_{i \neq j} |\langle f_i, f_j \rangle| \geq \sqrt{\frac{N-d}{d(N-1)}}, \quad N \geq d.
\end{equation}
The quantity $\alpha = \sqrt{\frac{N-d}{d(N-1)}}$ appearing on the right side of (\ref{WelchInEq}) is known as the \textit{Welch bound}. Sets of unit vectors attaining the lower bound in (\ref{WelchInEq}) are mathematical objects called \textit{equiangular tight frames} (ETFs) \cite{SH1, DattaHC}.
Such sets arise in many different areas as in communications, quantum information processing, and coding theory \cite{W1, MM1, DS1, SH1,Ren04, Sco06, KR1, RoyScott07, Hoggar82,Xia05}.
Consequently, the problem of constructing ETFs and determining conditions under which they exist has gained substantial attention \cite{SH1, Holmes04, Tropp05, Sustik07, Bodmann09, Bodmann2010, Waldron09, Fickus2012,redmond2009existence}. The Gram matrix of an ETF has two distinct eigenvalues: zero and $\frac Nd$ with multiplicities $N-d$ and $d,$ respectively \cite{Sustik07}. Conditions on eigenvalues for the existence of an ETF have been discussed in \cite{SH1,Holmes04, Sustik07, Bodmann09, Bodmann2010}, among others. A graph theoretic approach to constructing ETFs has been studied in \cite{Waldron09}. A correspondence discovered by Fickus et al.~\cite{Fickus2012} uses Steiner systems to directly construct the frame vectors of certain ETFs, bypassing the common technique of constructing a suitable Gram matrix.
This approach lets one construct highly redundant sparse ETFs.
However, in the real case, this approach can give rise to ETFs only if a real Hadamard matrix of a certain size exists.

Despite the desirability and importance of ETFs, these cannot exist for many pairs $(N, d).$
When the Hilbert space is $\mathbb{R}^d,$ the maximum number of equiangular lines is bounded by $\frac{d(d+1)}{2},$ and for $\mathbb{C}^d$ the bound is $d^2$ \cite{Lemmens73, DGS75}.
Even when these restrictions hold, ETFs are very hard to construct and do not exist for many pairs $(N,d)$ \cite{Sustik07}. This leads to generalizations and approximations of  ETFs.
For a real ETF, the off-diagonal entries of the Gram matrix are either $\alpha$ or $-\alpha$, where $\alpha$ is the Welch bound. In other words, the off-diagonal entries of the Gram matrix all have modulus equal to $\alpha.$
Generalizing this notion, a unit norm tight frame (UNTF) whose associated Gram matrix has off-diagonal entries with $k$ distinct \textit{moduli} is called a \textit{$k$-angle tight frame} \cite{Dattaa}.
These objects have also been explored in~\cite{casazza2017toward} under the name \textit{$d$-angular frames}.
Note that for a set of unit vectors, the diagonal entries of the corresponding Gram matrix will equal $1$.
Under this definition, ETFs are viewed as $1$-angle tight frames.
$2$-angle tight frames, or \textit{biangular tight frames}, are discussed in~\cite{casazza2017toward,haas2017constructions} and combinatorial constructions are provided.
It is to be noted that often in the literature, a unit norm tight frame is called a \textit{two-distance tight frame} \cite{Barg2014, Larman77} if the off-diagonal entries of the associated Gram matrix take on either of two values $a$ and $b.$ 
In that case, real ETFs are thought of as two-distance tight frames instead of $1$-angle tight frames, as done here.
Non-equiangular two-distance tight frames are examples of the biangular tight frames mentioned above.
Explicit constructions of $k$-angle tight frames can be found in \cite{Dattaa}.
Besides generalizing the notion of an ETF, $k$-angle tight frames prove to be important due to their connection to graphs and association schemes as discussed in \cite{Dattaa}.

The goal here is to come up with objects that can be considered approximations of ETFs.
In particular, one seeks UNTFs with low coherence among the vectors.
Section~\ref{section::supermaximal_approx_of_ETFs} contains the main contribution of the work presented here. In Section~\ref{section::supermaximal_approx_of_ETFs}, a deterministic way of constructing low coherence UNTFs is given that has the added benefit of being applicable when an equiangular set of lines, and in particular an ETF, of a certain size is known not to exist.
The idea is to start from an ETF of $d+1$ vectors in $\mathbb{R}^d$  and determine an \textit{optimal} set of vectors that can be added to this ETF such that the resulting union is a UNTF in $\mathbb{R}^d$ for which the maximum coherence between distinct vectors can be minimized even if there does not exist a set of the same size for which the Welch bound is attained.  This UNTF can be considered as an approximation of an ETF of the same size in $\mathbb{R}^d$. The exact meaning of ``optimal'' in this context is given in Definition~\ref{definition::optimal-frames-to-add-to-etf}. This approach relies on the fact that an ETF of $d+1$ vectors in $\mathbb{R}^d$ or $\mathbb{C}^d$ always exists \cite{SH1, Sustik07, Dattaa, Goyal_2001} which can be viewed as the vertices of a regular simplex centered at the origin. Combinatorial objects called block designs are used to determine the optimal sets to be added to the starting ETF.
The main results of Section~\ref{section::supermaximal_approx_of_ETFs} are \cref{theorem::hadamard-design-construction} and Theorem~\ref{theorem:: 2-designs : maximum cross correlation}. Section~\ref{section::conclusion} provides some concluding remarks.

\subsection{Notation and Preliminaries} \label{Notation}
In a finite dimensional Hilbert space like $\mathbb{R}^d$ or $\mathbb{C}^d,$ a spanning set is called a \textit{frame}.\footnote{In an infinite dimensional space, the notion of a frame is far more subtle \cite{Chr03, Dau92} and will not be needed here.}
Given a set $\{f_1, \ldots, f_N\}$ of vectors in $\mathbb{R}^{d}$ or $\mathbb{C}^d,$ let $F$ be the matrix whose columns are the vectors $f_1, \ldots, f_N.$
$F$ will be called the \textit{synthesis operator} of $\{f_{1},\dots,f_{N}\}$.
For a \textit{tight frame} the $d \times d$ matrix $F F^*$ is a multiple of the identity.
That is, $FF^{*}$ equals $\lambda I$, where $I$ is the identity, and $\lambda$ is called the \textit{frame bound}.
A tight frame for which each vector is unit norm is called a \textit{unit norm tight frame} (UNTF).
The matrix $F^*F$ is the \textit{Gram matrix} of the set $\{f_1, \ldots, f_N\}$ and its non-zero eigenvalues are the same as the eigenvalues of $FF^*.$
The $(i,j)^{\text{th}}$ entry of the Gram matrix is the inner product $\langle f_{j},f_{i}\rangle$.
An \textit{equiangular tight frame} (ETF) is a set $\{f_1, \ldots, f_N\}$ in a $d$-dimensional Hilbert space $\mathcal{H}$ satisfying \cite{Sustik07}:
\begin{enumerate}[(i)]
	\item $FF^{*} = \frac{N}{d}I$, i.e., the set is a tight frame.

	\item $\|f_i\| = 1,$ for $i = 1, \ldots, N,$ i.e., the set is unit norm.

	\item $|\langle f_i, f_j \rangle| = \alpha,$ $1\leq i\neq j \leq N,$ where $\alpha$ is the Welch bound.
\end{enumerate}
Throughout, $\mathcal{H}$ will be either $\mathbb{C}^d$ or $\mathbb{R}^d$.\footnote{The results can be easily generalized to any $d$-dimensional Hilbert space $\mathcal{H}$ since $\mathcal{H}$ would be isomorphic to $\mathbb{R}^d$ or $\mathbb{C}^d.$} A frame of $N$ vectors in $\mathbb{R}^{d}$ (respectively, $\mathbb{C}^{d}$) will be referred to as an $(N,d)$ real (respectively, complex) frame.
When $\mathcal{H}$ is not specified, the frame will  be called an $(N, d)$ frame.
Let $\set{\lambda_{i}}_{i=1}^{N}$ be the eigenvalues of the corresponding Gram matrix $G$.
The \textit{frame potential}~\cite{BF03} of a set of unit vectors $\{f_{1},\dots,f_{N}\}$ is the quantity $FP(\{f_{i}\}_{i=1}^{N})$ given by
\begin{equation}\label{equation :: frame potenial}
	FP(\{f_{i}\}_{i=1}^{N}) = \sum_{i=1}^{N}\sum_{j=1}^{N}|\langle f_{i},f_{j}\rangle|^{2} = \sum_{i=1}^{N}\lambda_{i}^{2} = \tr G^{2}.
\end{equation}
The following result on the frame potential is found in~\cite{BF03}.
\begin{theorem}[Theorem 6.2, \cite{BF03}]\label{theorem::untfs-minimize-frame-potential}
	Let $d, N \in \NN$ with $d\leq N$ and let $\veccoll{\vec{f}}{i}{i=1}{N}$ be a set of unit norm vectors in $\RR^{d}$ or $\CC^{d}$.
	Then $FP(\veccoll{\vec{f}}{i}{i=1}{N})$ is bounded below by $\frac{N^{2}}{d}$ with equality if and only if the frame is a unit norm tight frame (an orthonormal basis in the case $N=d$).
\end{theorem}
\section{Constructing low coherence UNTFs with block designs}\label{section::supermaximal_approx_of_ETFs}
\subsection{Optimal sets to add to ETFs}

As already mentioned in \cref{subsection :: background and motivation}, ETFs are useful due to their minimal maximum coherence between distinct vectors. Unfortunately, even though $(d+1, d)$ ETFs always exist,  an $(N,d)$ ETF may not always exist if $N>d+1$. Since ETFs already minimize the maximum coherence among pairs of its vectors, a natural approach to approximating an $(N,d)$ ETF, when $N>d+1$, is to add an \textit{optimal} set of vectors to a $(d+1,d)$ ETF in such a way that the resulting set is still a UNTF with maximum coherence among its distinct vectors as close to the Welch bound as possible. By definition, a $(d+1, d)$ ETF is a UNTF. Since the union of two UNTFs is another UNTF for the same space, the optimal set to be added will be taken to be a UNTF to ensure that the resulting set is a UNTF. This is accomplished below in \cref{theorem::hadamard-design-construction} and \cref{theorem:: 2-designs : maximum cross correlation}.
In the absence of an $(N,d)$ ETF, the low coherence UNTFs provided by \cref{theorem::hadamard-design-construction} and \cref{theorem:: 2-designs : maximum cross correlation} can be thought of as approximations of a corresponding $(N,d)$ ETF, if it were to exist.

From (\ref{WelchInEq}), for a $(d+1,d)$ ETF, the modulus of the inner product of any two distinct vectors is $\frac{1}{d}$.
In what follows, $\veccoll{\vec{f}}{i}{i=1}{d+1}\subset\RR^{d}$ will denote a $(d+1,d)$ ETF where $\dotprod{\vec{f}_{i}}{\vec{f}_{j}}=-\frac{1}{d}$ for $i\neq j$. For any $d$, this ETF can be constructed based on an explicit construction given in~\cite{Dattaa}.
\begin{definition}\label{definition::optimal-frames-to-add-to-etf}
	Let $\mathcal{G}_{0}\subset\RR^{d}$ be a UNTF. Let $\veccoll{\vec{f}}{i}{i=1}{d+1}\subset\RR^{d}$ be a $(d+1,d)$ ETF where $\dotprod{\vec{f}_{i}}{\vec{f}_{j}}=-\frac{1}{d}$ for $i\neq j$.
	$\mathcal{G}_{0}$ is said to be \textit{optimal with respect to ${\set{\vec{f}_{i}}_{i=1}^{d+1}}$} if, among all possible UNTFs $\mathcal{G}\subset\RR^{d}$, $\mathcal{G}_{0}$ minimizes the maximum coherence between distinct vectors of $\veccoll{\vec{f}}{i}{i=1}{d+1}\cup\mathcal{G}$.
\end{definition}
%
Given a set of unit vectors containing the real ETF $\set{f_{i}}_{i=1}^{d+1}$ of Definition~\ref{definition::optimal-frames-to-add-to-etf}, \cref{theorem::minimal-worst-case-for-adding-to-etf} below gives a lower bound for the maximum coherence between distinct vectors of this set. Using \cref{theorem::minimal-worst-case-for-adding-to-etf}, one can determine a set that is optimal with respect to $\set{f_{i}}_{i=1}^{d+1}$ in the sense of Definition~\ref{definition::optimal-frames-to-add-to-etf}.
The bound in (\ref{equation :: minimal worst case for adding to etf}) is similar to the orthoplex bound~\cite{bodmann2015achieving}.
\begin{theorem}\label{theorem::minimal-worst-case-for-adding-to-etf}
	Let $\{\vec{f}_{i}\}_{i=1}^{d+1}\subset\mathbb{R}^{d}$ denote an ETF satisfying $\langle \vec{f}_{i},\vec{f}_{j}\rangle =-\frac{1}{d}$ for $1\leq i<j\leq d+1$.
	Let $\vec{f}\in\mathbb{R}^{d}$ be a unit vector.
	Then
	\begin{equation}\label{equation :: minimal worst case for adding to etf}
		\max_{\mathclap{1\leq i\leq d+1}}|\langle \vec{f},\vec{f}_{i}\rangle| \geq \frac{1}{\sqrt{d}}.
	\end{equation}
\end{theorem}	
\begin{proof}
Since $\vec{f}$ is a unit vector and $\{\vec{f}_{i}\}_{i=1}^{d+1}$ is a UNTF,
	\[
		\frac{d+1}{d} = \frac{d+1}{d}\norm{f}^{2} = \sum_{i=1}^{d+1}|\langle \vec{f},\vec{f}_{i}\rangle|^{2}.
	\]
The sum $\sum_{i=1}^{d+1}|\langle \vec{f},\vec{f}_{i}\rangle|^{2}$ must be bounded above by $(d+1)\max_{i}|\langle \vec{f},\vec{f}_{i}\rangle|^{2}$, which gives
	\[
		\frac{d+1}{d} \leq (d+1)\max_{i}|\langle \vec{f},\vec{f}_{i}\rangle|^{2}.
	\]
	Simplifying this gives
	\[
		\frac{1}{\sqrt{d}}\leq\max_{i}|\langle \vec{f},\vec{f}_{i}\rangle|.
	\]
\end{proof}
\begin{remark}\label{remark::minimal-worst-case-for-adding-to-etf}
	Note that \cref{theorem::minimal-worst-case-for-adding-to-etf} only relies on the fact that $\{f_{i}\}_{i=1}^{d+1}$ is a UNTF.
	Hence the result still holds if the given $(d+1,d)$ ETF is replaced with an arbitrary $(N,d)$ UNTF, although this fact will not be required in this paper.
	It is worthwhile to note here that if $d$ is odd then~\cref{theorem::k-angle-near-optimality} shows that the inequality in~\cref{theorem::minimal-worst-case-for-adding-to-etf} is sharp for the $(d+1,d)$ ETF given in the statement of this theorem.
\end{remark}
%
\noindent
If $\mathcal{G}_{0}$ is a UNTF such that the maximum coherence between distinct vectors of $\mathcal{G}_{0}\cup\set{f_{i}}_{i=1}^{d+1}$ is $\frac{1}{\sqrt{d}}$, then \cref{theorem::minimal-worst-case-for-adding-to-etf} implies that $\mathcal{G}_{0}$ is optimal with respect to $\set{f_{i}}_{i=1}^{d+1}$. It is possible to obtain such UNTFs using the construction given in \cref{theorem::k-angle-near-optimality} below. This is then further developed in \cref{theorem::hadamard-design-construction}, \cref{theorem:: 2-designs : UNTFs}, and \cref{theorem:: 2-designs : maximum cross correlation}.
\begin{theorem}\label{theorem::k-angle-near-optimality}
	Let $d\in\NN$, set $k = \ceil*{\frac{d+1}{2}}$.
	Let $\{\vec{f}_{i}\}_{i=1}^{d+1}$ denote the $(d+1, d)$ ETF satisfying $\langle \vec{f}_{i},\vec{f}_{j}\rangle = -\frac{1}{d}$ for $i\neq j$.
	Let $\Lambda$ denote a subset of $\set{1,\dots,d+1}$ of size $k$, and define the vector $g$ by
	\[
		g = \frac{\sum_{i\in\Lambda}f_{i}}{\norm{\sum_{i\in\Lambda}f_{i}}}.
	\]
	Then
	\[
		\max_{\mathclap{1\leq j\leq d+1}}|\langle \vec{g},\vec{f}_{j}\rangle| = \frac{1}{\sqrt{d}} \quad\text{if $d$ is odd,}
	\]
	and
	\[
		\max_{\mathclap{1\leq j\leq d+1}}|\langle \vec{g},\vec{f}_{j}\rangle| \leq \sqrt{\frac{d+2}{d^{2}}} \quad\text{if $d$ is even.}
	\]
\end{theorem}
\begin{proof}
	First, note that
	\[
		\norm*{\sum_{i\in\Lambda}f_{i}}^{2}
		= \dotprod{\sum_{i\in\Lambda}f_{i}}{\sum_{j\in\Lambda}f_{j}} \\
		= k - \frac{k(k-1)}{d}.
	\]
	Hence $\norm{\sum_{i\in\Lambda}f_{i}} = \sqrt{\frac{k(d+1-k)}{d}}$, and it follows that
	\begin{align*}
		\dotprod{\vec{g}}{\vec{f}_{j}}
		&= \sqrt{\frac{d}{k(d+1-k)}}\left\langle \sum_{i\in\Lambda}\vec{f}_{i},\vec{f}_{j}\right\rangle\\
		&= \begin{cases}
			\sqrt{\frac{d+1-k}{dk}}&\text{if $j=i$ for some $i\in\Lambda$,} \\
			-\sqrt{\frac{k}{d(d+1-k)}}&\text{otherwise.}
		\end{cases}
	\end{align*}
	Now consider the case $k = \ceil*{\frac{d+1}{2}}$.
	If $d$ is odd then $k = \frac{d+1}{2}$ and
	$\max_{j}\abs*{\dotprod{\vec{g}}{\vec{f}_{j}}} = \frac{1}{\sqrt{d}}.$
	If $d$ is even, then $k = \frac{d+2}{2}$ and
	\[
		\langle g, f_{j}\rangle = \begin{cases}
			\frac{1}{\sqrt{d+2}} &\text{if $j=i$ for some $i\in\Lambda$,} \\
			-\sqrt{\frac{d+2}{d^{2}}} &\text{otherwise.}	
		\end{cases}
	\]
	Since $\frac{1}{\sqrt{d+2}}\leq\sqrt{\frac{d+2}{d^{2}}}$, $\max_{j}\abs*{\dotprod{\vec{g}}{\vec{f}_{j}}} \leq \sqrt{\frac{d+2}{d^{2}}}$ if $d$ is even.
\end{proof}	
\cref{theorem::k-angle-near-optimality} shows how to construct vectors $g$ whose maximum coherence with vectors of  the given ETF $\set{f_{i}}_{i=1}^{d+1}$ either meets or comes very close to meeting the optimal bound $\frac{1}{\sqrt{d}}$ of \cref{theorem::minimal-worst-case-for-adding-to-etf}. In order to construct UNTFs that are optimal with respect to the ETF $\set{f_{i}}_{i=1}^{d+1}\subset\RR^{d}$ satisfying $\dotprod{f_{i}}{f_{j}} = -\frac{1}{d}$ for $i\neq j$, one can then consider UNTFs of the form $\set{g_{i}}\subset\mathbb{R}^{d}$, where by the proof of \cref{theorem::k-angle-near-optimality},
\begin{equation}\label{equation :: k-angle vectors}
	g_{i} = \frac{\sum_{j\in\Lambda_{i}}f_{j}}{\norm{\sum_{j\in\Lambda_{i}}f_{j}}} = \sqrt{\frac{d}{k(d+1-k)}}\sum_{j\in\Lambda_{i}}f_{j},
\end{equation}
and $\set{\Lambda_{i}}$ is a collection of subsets of $\set{1,\dots,d+1}$ of size $k = \ceil{\frac{d+1}{2}}$.
The problem now is to determine how to select the subsets $\set{\Lambda_{i}}$ so that the vectors $\set{g_{i}}\subset\mathbb{R}^{d}$ satisfy the following constraints:
\begin{enumerate}[(i)]
	\item the maximum coherence of $\set{g_{i}}\cup\set{f_{i}}_{i=1}^{d+1}$ is $\frac{1}{\sqrt{d}}$,
	\item the set $\set{g_{i}} $ is a UNTF.
\end{enumerate}
Note that the constraint in (ii) is equivalent to $\set{g_{i}}\cup\set{f_{i}}_{i=1}^{d+1}$ being a UNTF since the union of two UNTFs is another UNTF.
To see how to choose subsets $\set{\Lambda_{i}}$ that minimize the maximum coherence, the following computations will be useful.
Recall that, for $i\neq j$, $\abs{\dotprod{f_{i}}{f_{j}}} = \frac{1}{d}$.
Also, due to \cref{theorem::minimal-worst-case-for-adding-to-etf,theorem::k-angle-near-optimality}, $\frac{1}{\sqrt{d}} \leq \max_{i,j}\abs*{\dotprod{g_{i}}{f_{j}}} \leq \sqrt{\frac{d+2}{d^{2}}}$.
It remains to determine $\max_{i \neq j}\abs{\dotprod{g_{i}}{g_{j}}}$.
Due to (\ref{equation :: k-angle vectors}),
\[
	|\langle \vec{g}_{i},\vec{g}_{j}\rangle| = \frac{d}{k(d+1-k)}\left| \sum_{m\in\Lambda_{i}}\sum_{n\in\Lambda_{j}}\dotprod{f_{m}}{f_{n}}\right|.
\]
For $i \neq j,$ let $l = |\Lambda_{i}\cap\Lambda_{j}|$.
Then out of the $k^{2}$ terms in the double summation, there will be $l$ terms that equal $1$ and $k^{2}-l$ terms that equal $-\frac{1}{d}$.
Therefore, for $i \neq j,$
\begin{equation}\label{equation:: k-angle cross correlation : number of intersections}
	 |\langle \vec{g}_{i},\vec{g}_{j}\rangle| = \frac{d}{k(d+1-k)}\left| l - \frac{1}{d}(k^{2}-l)\right| = \frac{d+1}{k(d+1-k)}\abs*{l - \frac{k^{2}}{d+1}}.
\end{equation}
Thus if $l$ is the closest integer to $\frac{k^{2}}{d+1}$, then the above inner product is minimized.
The next example shows this approach in action.
\begin{example}[An optimal UNTF to add to a $(4,3)$ ETF]\label{example:: optimal UNTF : R^3}
	Let $\{\vec{f}_{i}\}_{i=1}^{4}\subset\mathbb{R}^{3}$ be an ETF satisfying $\langle \vec{f}_{i},\vec{f}_{j}\rangle=-\frac{1}{3}$ for $1\leq i<j\leq 4$.
	Such an ETF can be constructed from the vertices of a regular tetrahedron inscribed within the unit sphere.
	Set $k = \ceil*{\frac{d+1}{2}} = 2$.
	Then the goal is to find a collection of subsets of $\{1,2,3,4\}$ of size $k=2$ such that the intersection of any two members has $l=\frac{k^{2}}{d+1} = 1$ element.
	One such collection is given by $\{\{1,2\}, \{1,3\}, \{1,4\}\}$.

	Now define $\{\vec{g}_{i}\}_{i=1}^{3}$ by
	\[
		\vec{g}_{1} = \frac{\vec{f}_{1}+\vec{f}_{2}}{\|\vec{f}_{1}+\vec{f}_{2}\|},\ \vec{g}_{2} = \frac{\vec{f}_{1}+\vec{f}_{3}}{\|\vec{f}_{1}+\vec{f}_{3}\|},\ \vec{g}_{3} = \frac{\vec{f}_{1}+\vec{f}_{4}}{\|\vec{f}_{1}+\vec{f}_{4}\|}.
	\]
	It can be checked that $\set{\vec{g}_{i}}_{i=1}^{3}$ forms an orthonormal basis for $\mathbb{R}^{3}.$ This implies that $\{\vec{f}_{i}\}_{i=1}^{4}\cup\{\vec{g}_{i}\}_{i=1}^{3}$
	is a $(7,3)$ UNTF. By \cref{theorem::k-angle-near-optimality},  the maximum coherence between distinct vectors of the set $\{\vec{f}_{i}\}_{i=1}^{4}\cup\{\vec{g}_{i}\}_{i=1}^{3}$  is $\frac{1}{\sqrt{3}}$ which is the lower bound given in \cref{theorem::minimal-worst-case-for-adding-to-etf}.
	By (\ref{WelchInEq}), the corresponding Welch bound for a collection of $7$ unit vectors in $\RR^{3}$ is $\sqrt{\frac{7-3}{3\cdot6}} = \frac{\sqrt{2}}{3} \approx .4714,$ whereas $\frac{1}{\sqrt{3}}\approx.5774$.
	However, there does not exist a $(7,3)$ ETF. In fact, the largest ETF in $\RR^{3}$ is a $(6,3)$ ETF, and so no collection of $7$ unit vectors in $\RR^{3}$ has maximum coherence $\frac{\sqrt{2}}{3}$. Further, by \cref{theorem::minimal-worst-case-for-adding-to-etf}, the maximum coherence between distinct vectors of the set $\{\vec{f}_{i}\}_{i=1}^{4}\cup\{\vec{g}_{i}\}_{i=1}^{3}$ is the smallest possible among all UNTFs containing the starting $(4, 3)$ ETF $\{\vec{f}_{i}\}_{i=1}^{4}.$ In the absence of a $(7,3)$ ETF, the set
$\{\vec{f}_{i}\}_{i=1}^{4}\cup\{\vec{g}_{i}\}_{i=1}^{3}$ can be thought of as a low coherence UNTF of $7$ vectors in $\mathbb{R}^3$ that is an approximation of the hypothetical $(7,3)$ ETF.
\end{example}
\noindent The set $\{\vec{f}_{i}\}_{i=1}^{4}\cup\{\vec{g}_{i}\}_{i=1}^{3}$  obtained in Example~\ref{example:: optimal UNTF : R^3} formed a tight frame due to the fact that $\{\vec{g}_{i}\}_{i=1}^{3}$ turned out to be an orthonormal basis and hence a UNTF. However, this is not guaranteed in general. 
One way to ensure tightness is by utilizing \textit{block designs}~\cite{moore2013difference}, as described in the following subsection.
\subsection[Using block designs to determine optimal sets to add to ETFs]{Using block designs to determine optimal sets to add to ETFs}\label{subsection :: using block designs}
\begin{definition}\label{definition:: 2-designs}
	Let $X$ denote a set containing $v$ points and suppose there is a collection $\mathcal{B}$ of subsets (``blocks'') of $X$ where each block has size $k$.
	If for any $x\in X$ there are precisely $r$ blocks in $\mathcal{B}$ containing $x$, and for any distinct $x,y\in X$ there are precisely $\lambda$ blocks containing $\{x,y\}$, then $\mathcal{B}$ is said to be a $(v,k,\lambda)$ \textit{block design on $X$}, or more simply a block design.
\end{definition}
\noindent Example~\ref{example:: optimal UNTF : R^3} indirectly makes use of \textit{symmetric designs} and \textit{Hadamard designs}~\cite{colbourn2010crc}, which are particular examples of block designs \footnote{Particular block designs known as \textit{Steiner systems} have been used to construct equiangular tight frames~\cite{Fickus2012}.}. This is explained below. 
\begin{definition}\label{definition::symmetric-hadamard-designs}
	Let $\mathcal{B}$ denote a $(v,k,\lambda)$ block design.
	$\mathcal{B}$ is a symmetric design if $|\mathcal{B}| = v$.
	$\mathcal{B}$ is a Hadamard design if it is a symmetric design and the parameters $v,k,\lambda$ satisfy $v = 4n-1, k = 2n-1$ and $\lambda = n-1$ or $v = 4n-1, k = 2n$ and $\lambda = n$.
\end{definition}
\noindent A $(4n-1, 2n-1, n-1)$ or $(4n-1,2n,n)$ Hadamard design exists if and only if there exists a corresponding $4n\times4n$ real Hadamard matrix~\cite{colbourn2010crc}.
Hadamard designs, like all symmetric designs, satisfy the following important property.
\begin{lemma}[\cite{colbourn2010crc}]\label{lemma:symmetric-design-intersection-number}
	Let $\mathcal{B}$ denote a $(v,k,\lambda)$ symmetric design.
	Then any two distinct blocks in $\mathcal{B}$ have precisely $\lambda$ elements in common. 
\end{lemma}
\noindent Hadamard designs, specifically $(4n-1,2n-1,n-1)$ Hadamard designs, can now be used to extend the construction given in Example~\ref{example:: optimal UNTF : R^3} to other dimensions besides $d=3$.
\begin{theorem}\label{theorem::hadamard-design-construction}
	Let $d = 4n-1$ for some $n$ and suppose that there exists a $(4n-1,2n-1,n-1) = (d, \frac{d+1}{2} - 1, \frac{d+1}{4} - 1)$ Hadamard design $\mathcal{B} = \{B_{i}\}_{i=1}^{d}$ on the set $\{2,3,\ldots,d+1\}$.
	Let $\{f_{i}\}_{i=1}^{d+1}\subset\RR^{d}$ denote an ETF satisfying $\langle f_{i},f_{j}\rangle = -\frac{1}{d}$ for $1\leq i<j\leq d+1$.
	Define $\Lambda_{i} = \{1\}\cup B_{i}$ for $1\leq i\leq d$ and for each $i$ construct the vector $g_{i}$ by
	\[
		g_{i} = \frac{\sum_{j\in\Lambda_i}f_{j}}{\norm{\sum_{j\in\Lambda_i}f_{j}}} . 
	\]
	Then $\{f_{i}\}_{i=1}^{d+1}\cup\{g_{i}\}_{i=1}^{d}$ is a $(2d+1,d)$ UNTF, and the maximum coherence between distinct vectors of this UNTF is $\frac{1}{\sqrt{d}}$.
\end{theorem}
\begin{proof}
	By~(\ref{equation:: k-angle cross correlation : number of intersections}), if $1\leq i<j\leq d$ then
	\[
		\abs{\langle g_{i}, g_{j}\rangle} = \frac{d+1}{k(d+1-k)}\abs*{l - \frac{k^{2}}{d+1}}
	\]
	where $k = |\Lambda_{i}| = \frac{d+1}{2}$ and $l = |\Lambda_{i}\cap\Lambda_{j}| = \frac{d+1}{4}$.
	Since
	\[
		l - \frac{k^{2}}{d+1} = \frac{d+1}{4} - \frac{d+1}{4} = 0,
	\]
	it follows that $\{g_{i}\}_{i=1}^{d}$ is an orthogonal set. Further, since each $g_i$ has unit norm by construction, $\{g_{i}\}_{i=1}^{d}$ is an orthonormal basis.
	Therefore $\{f_{i}\}_{i=1}^{d+1}\cup\{g_{i}\}_{i=1}^{d}$ must be a UNTF.
	The proof of~\cref{theorem::k-angle-near-optimality} also shows that $\max_{i,j}|\langle g_{i},f_{j}\rangle| = \frac{1}{\sqrt{d}}$.
	Hence $\{f_{i}\}_{i=1}^{d+1}\cup\{g_{i}\}_{i=1}^{d}$ is a $(2d+1,d)$ UNTF with maximum coherence between distinct vectors given by $\frac{1}{\sqrt{d}}$.
\end{proof}
\noindent
One can now observe that Example~\ref{example:: optimal UNTF : R^3} can also be obtained from \cref{theorem::hadamard-design-construction} by setting $n = 1.$ The needed $(3, 1, 0)$ Hadamard design on the set $\{2, 3, 4\}$ can be taken as $B_1 = \{2 \},$ $B_2 = \{ 3 \},$ $B_3 = \{4\}.$
\begin{example}[An optimal UNTF to add to an $(8,7)$ ETF]
	Let $\{f_{i}\}_{i=1}^{8}\subset\RR^{7}$ be an ETF satisfying $\langle f_{i},f_{j}\rangle = -\frac{1}{7}$ for $1\leq i<j\leq 8$.
	Then \cref{theorem::hadamard-design-construction} shows that a $(7,3,1)$ Hadamard design on the set $\{2,3,\ldots,8\}$ can be used to find an optimal UNTF (specifically, an orthonormal basis) to add to this ETF.
	One such design can be found in~\cite{colbourn2010crc} and is given by
	\begin{align*}
		B_{1} &= \{2,5,6 \}, & 		B_{5} &= \{2,7,8 \}, \\
		B_{2} &= \{3,5,7 \}, &		B_{6} &= \{3,6,8 \}, \\
		B_{3} &= \{4,5,8 \}, &	    B_{7} &= \{4,6,7 \}. \\
		B_{4} &= \{2,3,4 \}, &	   & 
	\end{align*}
	The next step is to add $1$ to each block to create the sets $\{\Lambda_{i}\}_{i=1}^{7}$, and use these to construct the orthonormal basis $\{g_{i}\}_{i=1}^{7}$:
	\begin{align*}
		\vec{g}_{1} = \frac{\vec{f}_{1}+\vec{f}_{2}+\vec{f}_{5} + \vec{f}_{6}}{\norm*{\vec{f}_{1}+\vec{f}_{2}+\vec{f}_{5} + \vec{f}_{6}}},\ \dots,\ \vec{g}_{7} = \frac{\vec{f}_{1}+\vec{f}_{4}+\vec{f}_{6} + \vec{f}_{7}}{\norm*{\vec{f}_{1}+\vec{f}_{4}+\vec{f}_{6} + \vec{f}_{7}}}.
	\end{align*}
	The resulting $(15,7)$ UNTF $\{f_{i}\}_{i=1}^{8}\cup\{g_{i}\}_{i=1}^{7}$ then has maximum coherence between distinct vectors given by $\frac{1}{\sqrt{7}}$.
\end{example}
\noindent
\cref{theorem::hadamard-design-construction} uses Hadamard designs to construct low coherence UNTFs that meet the lower bound of \cref{theorem::minimal-worst-case-for-adding-to-etf}. However, as mentioned previously these designs are equivalent to the existence of a real Hadamard matrix of the right size.
It is therefore desirable to use more general block designs to construct low coherence UNTFs.
This will require the following definition.
\begin{definition}\label{example:: 2-designs : incidence matrix}
	Let $X = \set{x_{i}}_{i=1}^{v}$ denote a finite set and let $\mathcal{B} = \set{B_{i}}_{i=1}^{b}$ denote a block design on $X$.
	The matrix $K = \brackets{k_{ij}}$ for $1\leq i\leq v$, $1\leq j\leq b$, where
	\[
		k_{ij} =
		\begin{cases}
			1 &\text{if $x_{i}\in B_{j}$,} \\
			0 &\text{otherwise,}
		\end{cases}
	\]
	is called the incidence matrix of $\mathcal{B}$.
\end{definition}
\noindent The following fundamental result on block designs will be necessary, and can be found in~\cite{moore2013difference}.
\begin{lemma}[\cite{moore2013difference}]\label{lemma:: 2-designs : constraints : number of blocks : incidence matrix : lambda}
	Let $X = \set{x_{i}}_{i=1}^{v}$ denote a finite set and let $\mathcal{B} = \set{B_{i}}_{i=1}^{b}$ denote a $(v,k,\lambda)$ block design on $X$.
	Let $r$ denote the number of blocks in $\mathcal{B}$ containing a given element of $X$.
	\begin{enumerate}[(i)]
		\item
		Then
		\[
			b = \frac{\lambda\binom{v}{2}}{\binom{k}{2}}\quad\text{and}\quad r(k-1) = \lambda(v-1).
		\]
	
		\item
		Let $K = [k_{ij}]$ denote the incidence matrix of $\mathcal{B}$.
		Then
		\[
			KK^{T} = (r-\lambda)I+\lambda J
		\]
		where $J$ denotes the $v\times v$ matrix whose entries are all $1$.
	\end{enumerate}
\end{lemma}
\noindent The following inequality will also be required and can be found in~\cite{bose1949}.
\begin{theorem}[Fisher's Inequality \cite{bose1949}]\label{THM:Fisher}
If $\mathcal{B}$ is a block design on $X = \{1, \ldots, d+1\}$ then $\mathcal{B}$ must contain at least $d+1$ blocks.
\end{theorem}
\begin{theorem}\label{theorem:: 2-designs : UNTFs}
	Let $\{\vec{f}_{i}\}_{i=1}^{d+1}\subset\mathbb{R}^{d}$ denote an ETF satisfying $\langle \vec{f}_{i},\vec{f}_{j}\rangle =-\frac{1}{d}$ for $1\leq i<j\leq d+1$, and suppose that $\mathcal{B} = \{B_{i}\}_{i=1}^{b}$ is a $(d+1,k,\lambda)$ block design on $\{1,\ldots,d+1\}$ for some $k,\lambda\in\NN$.
	Define $\vec{g}_{i}$ for $1\leq i\leq b$ by
	\[
		\vec{g}_{i} = \frac{\sum_{j\in B_{i}}\vec{f}_{j}}{\norm*{\sum_{j\in B_{i}}\vec{f}_{j}}}.
	\]
	Then $\set{\vec{g}_{i}}_{i=1}^{b}$ is a $(b, d)$ UNTF. 
\end{theorem}
\begin{proof}
	Let $F$ denote the synthesis operator of $\set{\vec{f}_{i}}_{i=1}^{d+1}$, let $G$ denote the corresponding Gram matrix, and let $K$ be the incidence matrix of the block design $\mathcal{B}$.
	Let $F_{1}$ denote the synthesis operator of $\set{\vec{g}_{i}}_{i=1}^{b}$ and $G_{1}$ the corresponding Gram matrix.
	Then using the proof of \cref{theorem::k-angle-near-optimality}, one can write
	\[
		F_{1} = \sqrt{\frac{d}{k(d+1-k)}}FK\quad\text{and}\quad G_{1} =  F_{1}^{T}F_{1} = \frac{d}{k(d+1-k)}K^{T}GK.
	\]
	To show that $\set{\vec{g}_{i}}_{i=1}^{b}$ is a UNTF, the frame potential is used in accordance to \cref{theorem::untfs-minimize-frame-potential}.
	This is valid since $b > d$ by~\cref{THM:Fisher}.
	Note that
	\begin{align*}
		FP(\set{\vec{g}_{i}}_{i=1}^{b})
		= \tr G_{1}^{2} = \left(\frac{d}{k(d+1-k)}\right)^{2}\tr(K^{T}GKK^{T}GK).
	\end{align*}
	By Lemma~\ref{lemma:: 2-designs : constraints : number of blocks : incidence matrix : lambda} (ii),
	\[
		KK^{T} = (r-\lambda)I+\lambda J.
	\]
	Since $GJ$ is the zero matrix,
	\begin{align*}
		\tr G_{1}^{2}
		&= \left(\frac{d}{k(d+1-k)}\right)^{2}\tr((r-\lambda)G^{2}KK^{T}) \\
		&= \left(\frac{d}{k(d+1-k)}\right)^{2}(r-\lambda)^{2}\frac{(d+1)^{2}}{d}
	\end{align*}
	where the last equality follows from the fact that the frame potential of the original ETF is $\frac{(d+1)^{2}}{d}$.
	By Lemma~\ref{lemma:: 2-designs : constraints : number of blocks : incidence matrix : lambda} (i),
	\[
		\lambda = \frac{bk(k-1)}{(d+1)d}\quad\text{and}\quad r = \frac{d\lambda}{k-1}.
	\]
	Thus
	\[
		r-\lambda = \frac{bk(d+1-k)}{(d+1)d},
	\]
	and so
	\begin{align*}
		\tr G_{1}^{2}
		= \left(\frac{d}{k(d+1-k)}\right)^{2}(r-\lambda)^{2}\frac{(d+1)^{2}}{d} = \frac{b^{2}}{d}.
	\end{align*}
	Therefore $\set{\vec{g}_{i}}_{i=1}^{b}$ is a $(b, d)$ UNTF by \cref{theorem::untfs-minimize-frame-potential}.
\end{proof}
\noindent
As (\ref{equation:: k-angle cross correlation : number of intersections}) shows, the coherence between two distinct vectors $\vec{g}_{i}$ and $\vec{g}_{j}$ obtained from this construction is related to the size of the intersection of the blocks $B_{i}$ and $B_{j}$ that determine $\vec{g}_{i}$ and $\vec{g}_{j}$.
\begin{definition}\label{definition:: 2-designs : intersection numbers}
	Let $\mathcal{B} = \{B_{i}\}_{i=1}^{b}$ denote a block design on a set $X$.
	An integer $n\geq0$ is said to be an \textit{intersection number} of $\mathcal{B}$ if there are blocks $B_{i}$ and $B_{j}$ such that $n = |B_{i}\cap B_{j}|$.
\end{definition}
\noindent Bounds on the possible intersection numbers of a block design are required in order to obtain UNTFs that are optimal in the sense of Definition~\ref{definition::optimal-frames-to-add-to-etf}.
The bound below was originally given in~\cite{majumdar1953some} but the form used here is from~\cite{Beutelspacher82}.
\begin{theorem}[\cite{Beutelspacher82}]\label{theorem:: intersection numbers : lower and upper bounds}
	Let $\mathcal{B} = \{B_{i}\}_{i=1}^{b}$ denote a $(v,k,\lambda)$ block design on a set $X$, and let $r$ denote the number of blocks containing a given element of $X$.
	Define $\sigma,\tau$ and $\Sigma$ as follows:
	\begin{align*}
		\sigma &= k-r+\lambda \\
		\tau &= \frac{k}{v}(2k-v) \\
		\Sigma &= \frac{2k\lambda}{r} - (k-r+\lambda).
	\end{align*}
	Let $B_{i}$ and $B_{j}$ denote distinct blocks of $\mathcal{B}$.
	Then
	\[
		\max\{\sigma,\tau\} \leq |B_{i}\cap B_{j}| \leq \Sigma.
	\]
\end{theorem}
\begin{remark}\label{remark:: intersection numbers : 0}
	Beutelspacher~\cite{Beutelspacher82} gives a slight refinement of the bound given in \cref{theorem:: intersection numbers : lower and upper bounds} by showing that if $\tau$ as defined in \cref{theorem:: intersection numbers : lower and upper bounds} is an intersection number of a block design, then $\tau$ must equal $0$.
\end{remark}
\noindent Finally, the previous results may now be used to obtain a UNTF $\set{g_{i}}_{i=1}^b$ that is optimal or nearly optimal with respect to the $(d+1, d)$ ETF $\set{f_{i}}_{i=1}^{d+1}$. Let $N = b + d+1.$
In cases where an $(N,d)$ ETF does not exist, the low coherence UNTF $\set{g_{i}}_{i=1}^b \cup\set{f_{i}}_{i=1}^{d+1} \subset \mathbb{R}^d$ can be considered an approximation.
This is presented in \cref{theorem:: 2-designs : maximum cross correlation}.
\begin{theorem}\label{theorem:: 2-designs : maximum cross correlation}
	Let $\veccoll{\vec{f}}{i}{i=1}{d+1}\subset\RR^{d}$ denote an ETF satisfying $\dotprod{\vec{f}_{i}}{\vec{f}_{j}} = -\frac{1}{d}$ for $1\leq i<j\leq d+1$.
	Let $\mathcal{B} = \{B_{i}\}_{i=1}^{b}$ denote a $(d+1,k,\lambda)$ block design on $\{1,\ldots,d+1\}$.
	Let $\set{\vec{g}_{i}}_{i=1}^{b}$ be given by
	\[
		\vec{g}_{i} = \frac{\sum_{j\in B_{i}}\vec{f}_{j}}{\norm*{\sum_{j\in B_{i}}\vec{f}_{j}}}.
	\]
	Then
	\begin{equation}\label[ineq]{inequality :: block design : cross correlation : lambda}
		\abs*{\dotprod{\vec{g}_{i}}{\vec{g}_{j}}} \leq \frac{d+1}{k(k-1)}\lambda-1.
	\end{equation}
	Furthermore, suppose $k = \ceil*{\frac{d+1}{2}}$ and $\mathcal{B}$ is a $(d+1,k,\lambda)$ block design with $\lambda = \ceil*{\frac{k(k-1)}{d+1}}$.
	Then $\set{f_{i}}_{i=1}^{d+1}\cup\set{g_{i}}_{i=1}^{b}$ is a $(d + 1 + b, d)$ UNTF satisfying the following.
	\begin{enumerate}[(i)]
		\item
		If $d$ is odd then the maximum coherence between distinct vectors of $\set{\vec{f}_{i}}_{i=1}^{d+1}\cup\set{\vec{g}_{i}}_{i=1}^{b}$ is bounded above by
		\[
			\max\set*{\frac{1}{\sqrt{d}},\frac{3}{d-1}},
		\]
and
		\item
		if $d$ is even then the maximum coherence between distinct vectors of $\set{\vec{f}_{i}}_{i=1}^{d+1}\cup\set{\vec{g}_{i}}_{i=1}^{b}$ is bounded above by
		\[
			 \max\set*{\sqrt{\frac{d+2}{d^{2}}},\frac{4}{d+2}+\frac{3}{d(d+2)}}.
		\]
	\end{enumerate}
\end{theorem}
\begin{proof}
	Let $\sigma$ and $\Sigma$ be as given in \cref{theorem:: intersection numbers : lower and upper bounds}.
	By Remark~\ref{remark:: intersection numbers : 0}, the bound in \cref{theorem:: intersection numbers : lower and upper bounds} becomes
	\[
		\sigma\leq |B_{i}\cap B_{j}| \leq \Sigma.
	\]
	First, recall (\ref{equation:: k-angle cross correlation : number of intersections}):
	\[
		\abs*{\dotprod{\vec{g}_{i}}{\vec{g}_{j}}} = \frac{d+1}{k(d+1-k)}\abs*{ l-\frac{k^{2}}{d+1}}.
	\]	
	Hence the intersection number of $\mathcal{B}$ that is \textit{farthest} from $\frac{k^{2}}{d+1}$ will give $\max_{i\neq j}\abs*{\dotprod{\vec{g}_{i}}{\vec{g}_{j}}}$.
	Let $\bar{\sigma} = \frac{\sigma+\Sigma}{2}$.
	It will be shown that $\bar{\sigma} \leq \frac{k^{2}}{d+1}$, which will imply that $\Sigma$ is closer to $\frac{k^{2}}{d+1}$ than $\sigma$.
	The definitions of $\sigma$ and $\Sigma$, as well as the relation $\frac{\lambda}{r} = \frac{k-1}{d}$ which is obtained from Lemma~\ref{lemma:: 2-designs : constraints : number of blocks : incidence matrix : lambda}, give
	\[
	 	\bar{\sigma} = \frac{\sigma+\Sigma}{2} = \frac{k\lambda}{r} = \frac{k(k-1)}{d}.
	\]
	Therefore,
	\begin{align*}
		\bar{\sigma}-\frac{k^{2}}{d+1}= k\left[\frac{k-1}{d}-\frac{k}{d+1}\right] = -k\left[\frac{d+1-k}{d(d+1)}\right] \leq0.
	\end{align*}
	Thus $l=\sigma$ gives the largest possible value in (\ref{equation:: k-angle cross correlation : number of intersections}) and so
	\[
		\abs*{\dotprod{\vec{g}_{i}}{\vec{g}_{j}}} \leq \frac{d+1}{k(d+1-k)}\left|\sigma - \frac{k^{2}}{d+1}\right|.
	\]
	Using the expression for $\sigma$ from \cref{theorem:: intersection numbers : lower and upper bounds} as well as the fact that $\sigma\leq\frac{k^{2}}{d+1}$,
	\begin{align*}
		\frac{d+1}{k(d+1-k)}\left|\sigma - \frac{k^{2}}{d+1}\right|
		&= \frac{d+1}{k(d+1-k)}\left[\frac{k^{2}}{d+1}-\sigma\right] \\
		&= \frac{k}{d+1-k}-\frac{d+1}{k(d+1-k)}[k-(r-\lambda)].
	\end{align*}
	As shown in the proof of \cref{theorem:: 2-designs : UNTFs},
	\[
		r-\lambda = \frac{bk(d+1-k)}{d(d+1)}
	\]
	which gives
	\[
		\frac{k}{d+1-k}-\frac{d+1}{k(d+1-k)}[k-(r-\lambda)] = \frac{d+1}{k(k-1)}\lambda-1
	\]
	after substituting the expression for $b$ from Lemma~\ref{lemma:: 2-designs : constraints : number of blocks : incidence matrix : lambda}.
	Therefore,
	\[
		 \abs*{\dotprod{\vec{g}_{i}}{\vec{g}_{j}}}\leq\frac{d+1}{k(k-1)}\lambda-1.
	\]
	
	By \cref{theorem:: 2-designs : UNTFs}, $\set{g_{i}}_{i=1}^{b}$ is a UNTF.
	Since $\set{f_{i}}_{i=1}^{d+1}$ is also a UNTF, $\set{f_{i}}_{i=1}^{d+1}\cup\set{g_{i}}_{i=1}^{b}$ is a UNTF as well.
	The rest of the proof involves finding the bounds of parts (i) and (ii).
	Now let $k = \ceil*{\frac{d+1}{2}}$ and suppose that $\lambda=\ceil{\frac{k(k-1)}{d+1}}$.
	If $d$ is odd, then $k = \frac{d+1}{2}$ and $\lambda = \ceil*{\frac{d-1}{4}}$.
	Thus $\lambda = \frac{d-1}{4}+\varepsilon$ where $0\leq\varepsilon\leq\frac{3}{4}$, and it follows that
	\begin{align*}
		\abs*{\dotprod{\vec{g}_{i}}{\vec{g}_{j}}}
		\leq \frac{d+1}{k(k-1)}\lambda-1= \frac{4}{d-1}\left[\frac{d-1}{4}+\varepsilon\right]-1 \leq \frac{3}{d-1}.
	\end{align*}
	Similarly, if $d$ is even then $k = \frac{d+2}{2}$ and $\lambda = \ceil*{\frac{d(d+2)}{4(d+1)}}$.
	Thus $\lambda = \frac{d(d+2)}{4(d+1)}+\varepsilon$ where $0\leq \varepsilon\leq \frac{4(d+1)-1}{4(d+1)}$ and so
	\begin{align*}
		 \abs*{\dotprod{\vec{g}_{i}}{\vec{g}_{j}}}\leq\frac{d+1}{k(k-1)}\lambda-1= \frac{4(d+1)\varepsilon}{d(d+2)} \leq \frac{4}{d+2}+\frac{3}{d(d+2)}.
	\end{align*}
	Combining these bounds with the bounds on the coherence between $\set{g_{i}}$ and $\set{f_{i}}$ given in \cref{theorem::k-angle-near-optimality} finishes the proof.
\end{proof}
\begin{remark}
The bounds in (i) and (ii) of \cref{theorem:: 2-designs : maximum cross correlation} are upper bounds on the maximum coherence between distinct vectors in the low coherence UNTF $\set{\vec{f}_{i}}_{i=1}^{d+1}\cup\set{\vec{g}_{i}}_{i=1}^{b}$. 
When $d$ is odd and $d\geq 11$, the bound (i) simplifies to $\frac{1}{\sqrt{d}}$.
Similarly, if $d$ is even and $d\geq12$, the bound (ii) simplifies to $\sqrt{\frac{d+2}{d^{2}}}$.
In particular, if $d$ is odd, $d\geq 11$ and there exists a $(d+1,k,\lambda)$ block design $\mathcal{B}$ with $k = \frac{d+1}{2}$ and $\lambda = \ceil{\frac{d-1}{4}}$, then the maximum coherence of the UNTF constructed using~\cref{theorem:: 2-designs : maximum cross correlation} meets the lower bound of~\cref{theorem::minimal-worst-case-for-adding-to-etf}.
\end{remark}
\noindent
According to the conditions given in Proposition 3.2 in~\cite{Holmes04}, a real $(22,10)$ ETF does not exist.
The following example uses \cref{theorem:: 2-designs : maximum cross correlation} to give a low coherence UNTF in this case.
\begin{example}[Approximating a $(22,10)$ ETF]\label{example:: symmetric design : optimal untf}
	Let $\{\vec{f}_{i}\}_{i=1}^{11}\subset\RR^{10}$ denote an ETF where $\dotprod{\vec{f}_{i}}{\vec{f}_{j}} = -\frac{1}{10}$ for $1\leq i<j\leq 11$.
	Let $k = \ceil*{\frac{11}{2}} = 6$ and $\lambda = \ceil*{\frac{6\cdot5}{11}} = 3$.
	Then the goal is to find an $(11,6,3)$ block design.
	One such design can be found in~\cite{cochran1957experimental} and is given by $\mathcal{B} = \set{B_{i}}_{i=1}^{11}$ with
	\begin{align*}
		B_{1} &= \{4,6,7,9,10,11\}, & 	B_{5} &= \{2,3,4,8,10,11\},	& B_{9} &= \{1,3,4,6,7,8\}, \\
		B_{2} &= \{1,5,7,8,10,11\}, &	B_{6} &= \{1,3,4,5,9,11\},	& B_{10} &= \{2,4,5,7,8,9\}, \\
		B_{3} &= \{1,2,6,8,9,11\}, &	    B_{7} &= \{1,2,4,5,6,10\},	& B_{11} &= \{3,5,6,8,9,10\}. \\
		B_{4} &= \{1,2,3,7,9,10\}, &	    B_{8} &= \{2,3,5,6,7,11\},	
	\end{align*}
\noindent Now define $\set{\vec{g}_{i}}_{i=1}^{11}$ by
	\begin{align*}
		\vec{g}_{1} = \frac{\vec{f}_{4}+\vec{f}_{6}+\dots+\vec{f}_{11}}{\norm*{\vec{f}_{4}+\vec{f}_{6}+\dots+\vec{f}_{11}}},\ \dots,\ \vec{g}_{11} = \frac{\vec{f}_{3}+\vec{f}_{5}+\dots+\vec{f}_{10}}{\norm*{\vec{f}_{3}+\vec{f}_{5}+\dots+\vec{f}_{10}}}.
	\end{align*}
	Then by \cref{theorem:: 2-designs : UNTFs,theorem:: 2-designs : maximum cross correlation}, $\{\vec{f}_{i}\}_{i=1}^{11}\cup\{\vec{g}_{i}\}_{i=1}^{11}$ is a UNTF with the maximum coherence between distinct vectors being bounded above by
	\[
		\max\set*{\sqrt{\frac{d+2}{d^{2}}},\frac{4}{d+2}+\frac{3}{d(d+2)}} 
		= \frac{43}{120}\approx.3583.
	\]
By \cref{theorem::minimal-worst-case-for-adding-to-etf}, the lower bound on the maximum coherence between distinct vectors for this example is $\frac{1}{\sqrt{10}}\approx.3162$. Thus, the the maximum coherence between distinct vectors lies in $[.3162, .3583].$ The Welch bound for a hypothetical $(22,10)$ ETF is $\sqrt{\frac{2}{35}}\approx.2390$. The low coherence UNTF $\{\vec{f}_{i}\}_{i=1}^{11}\cup\{\vec{g}_{i}\}_{i=1}^{11}$ obtained here can be thought of as an  approximation of a $(22,10)$ ETF.
\end{example}
\begin{remark}
	Armed with a table of block designs (such as the those provided in~\cite{cochran1957experimental,colbourn2010crc}), it is possible to use~\cref{theorem:: 2-designs : maximum cross correlation} to investigate a wide variety of candidates for low coherence UNTFs. 
	If $d$ is odd and $d\geq11$, and if one has a $(d+1,k,\lambda)$ block design with $k = \frac{d+1}{2}$ and $\lambda = \ceil*{\frac{k(k-1)}{d+1}}$, i.e., a block design that satisfies the constraints required to use the bound (i) in \cref{theorem:: 2-designs : maximum cross correlation}, then \cref{theorem:: 2-designs : maximum cross correlation} can be used to construct a low coherence UNTF that meets the lower bound of \cref{theorem::minimal-worst-case-for-adding-to-etf}.
	Although the authors did not come across block designs that satisfied the constraints required for the bound (i) in~\cref{theorem:: 2-designs : maximum cross correlation}, there are more designs that satisfy the bound (ii) than the one given in Example~\ref{example:: symmetric design : optimal untf}.
	One way to find these designs is to consider \textit{complements} of block designs.
	The incidence matrix of the complement design is obtained by swapping $0$ with $1$ in the incidence matrix of the original design.
	For example, the bound (ii) does not apply to any $(19,9,4)$ block design (which may be found in~\cite{colbourn2010crc}), but the complement of such a design has parameters $(19,10,5)$.
	These parameters do in fact satisfy the constraints required to use the bound (ii).
\end{remark}
\begin{example}\label{example::nontight_approx_ETF}
As seen in Example~\ref{example:: symmetric design : optimal untf}, when using block designs to find low coherence UNTFs in a certain dimensional space, one is constrained by a frame size that is determined by the block design used. For a lower redundancy, one might wish for a UNTF of smaller size.

Suppose one wants a real $(21, 10)$ ETF which does not exist by Corollary 16 in~\cite{Sustik07}. Consider the $(22, 10)$ UNTF of Example~\ref{example:: symmetric design : optimal untf}, and discard a vector. The resulting (21,10) frame is no longer tight, and the condition number is approximately $1.833$ for any vector removed. Note that for a tight frame the condition number is always one.
After removing any vector, the maximum coherence between distinct vectors of the resulting $(21,10)$ frame is bounded above by $\frac{43}{120}\approx .3583$, while the Welch bound for a hypothetical $(21,10)$ ETF is $\sqrt{\frac{11}{200}}\approx.2345.$


\end{example}
%
\section{Conclusion}\label{section::conclusion}
%
In a $d$-dimensional space, the UNTFs constructed in Section~\ref{section::supermaximal_approx_of_ETFs} all have maximum coherence bounded below by $\frac{1}{\sqrt{d}}$ due to Theorem~\ref{theorem::minimal-worst-case-for-adding-to-etf}.
In order to judge the effectiveness of a particular $(N,d)$ low coherence UNTF, one can compare this maximum coherence to the Welch bound of the corresponding hypothetical $(N,d)$ ETF.
The construction in Section~\ref{subsection :: using block designs} approximates an $(N,d)$ ETF by using block designs and adding a UNTF to an ETF of size $d+1$. In this process, the number of added vectors is at least $d+1$ due to Theorem~\ref{THM:Fisher}.
Thus the size $N$ of the low coherence UNTF resulting from this construction is $N\geq 2(d+1).$
The Welch bound $\sqrt{\frac{N-d}{d(N-1)}} = \sqrt{\frac 1d - \frac{d-1}{d(N-1)}}$ increases as $N$ increases and $d$ is held fixed. Thus, when $N\geq 2(d+1),$ the corresponding Welch bound is at least $\sqrt{\frac{d+2}{2d^{2}+d}}$.
For large $d$, this is approximately equal to $\frac{1}{\sqrt{2d}}$.
In comparison, a corresponding low coherence UNTF obtained from Theorem~\ref{theorem:: 2-designs : maximum cross correlation} using a $(d+1,k,\lambda)$ block design with $k=\ceil*{\frac{d+1}{2}}$ and $\lambda=\ceil*{\frac{k(k-1)}{d+1}}$ would have maximum coherence approximately equal to $\frac{1}{\sqrt{d}}$.
So for large $N$ and large $d$ one can expect the frames provided by using such block designs to have maximum coherence larger than the corresponding Welch bound by a factor of about $\sqrt{2}$.

For a given dimension $d$, the construction of low coherence UNTFs discussed in Section~\ref{section::supermaximal_approx_of_ETFs} starts with a $(d+1, d)$ ETF.  Instead, one can start with other sets of vectors, such as orthonormal bases or other ETFs. $(d+1,d)$ ETFs are slightly more redundant than orthonormal bases and are straightforward to construct in all finite dimensional Hilbert spaces, making them a desirable starting point for the construction in this paper, but it would still be worthwhile to investigate this construction starting with other sets of vectors.
Determining for which UNTFs the inequality provided in \cref{theorem::minimal-worst-case-for-adding-to-etf} is sharp (see Remark~\ref{remark::minimal-worst-case-for-adding-to-etf}) and an analogue of Theorem~\ref{theorem:: 2-designs : UNTFs} for other sets of vectors would be particularly interesting.

\section*{Acknowledgement}

The authors would like to thank the anonymous reviewer for insightful comments and useful suggestions that greatly helped to improve the quality of the paper. 

\section*{Funding}

This material is based upon work supported by the National Science Foundation under Award No. CCF-1422252

\bibliographystyle{unsrt}
\bibliography{ETF_LAA_refs}

\begin{thebibliography}{10}

\bibitem{W1}
L.~R. Welch.
\newblock Lower bounds on the maximum cross correlation of signals.
\newblock {\em IEEE Transactions on Information Theory}, 20(3):397--399, 1974.

\bibitem{SH1}
T.~Strohmer and R.~W. Heath, Jr.
\newblock Grassmannian frames with applications to coding and communication.
\newblock {\em Applied and Computational Harmonic Analysis}, 14(3):257--275,
  2003.

\bibitem{DattaHC}
S.~Datta, S.~D. Howard, and D.~Cochran.
\newblock Geometry of the {W}elch bounds.
\newblock {\em Linear Algebra Appl.}, 437(10):2455 -- 2470, 2012.

\bibitem{MM1}
J.~L. Massey and T.~Mittelholzer.
\newblock Welch's bound and sequence sets for code-division multiple-access
  systems.
\newblock In R.~Capocelli, A.~{De Santis}, and U.~Vaccaro, editors, {\em
  Sequences II: Methods in Communication, Security and Computer Science}, pages
  63--78. Springer-Verlag, New York, 1991.

\bibitem{DS1}
D.~V. Sarwate.
\newblock Meeting the {W}elch bound with equality.
\newblock In T.~Helleseth and H.~Niederreiter, editors, {\em Sequences and
  Their Applications}, DMTCS Series, chapter~11. Springer-Verlag, 1999.

\bibitem{Ren04}
J.~M. Renes, R.~Blume-Kohout, A.~J. Scott, and C.~M. Caves.
\newblock Symmetric informationally complete quantum measurements.
\newblock {\em Journal of Mathematical Physics}, 45(6):2171 -- 2180, 2004.

\bibitem{Sco06}
A.~J. Scott.
\newblock Tight informationally complete quantum measurements.
\newblock {\em Journal of Physics A: Mathematical and General}, 39:13507 --
  13530, 2006.

\bibitem{KR1}
A.~Klappenecker and M.~R{\"o}tteler.
\newblock Mutually unbiased bases are complex projective 2-designs.
\newblock In {\em Proceedings of the International Symposium on Information
  Theory}, pages 1740 -- 1744, September 2005.

\bibitem{RoyScott07}
A.~Roy and A.~J. Scott.
\newblock Weighted complex projective $2$-designs from bases: {O}ptimal state
  determination by orthogonal measurements.
\newblock {\em Journal of Mathematical Physics}, 48(072110), 2007.

\bibitem{Hoggar82}
S.~G. Hoggar.
\newblock $t$-designs in projective spaces.
\newblock {\em European J. Combin.}, 3(3):233--254, 1982.

\bibitem{Xia05}
P.~Xia, S.~Zhou, and G.~B. Giannakis.
\newblock Achieving the {W}elch bound with difference sets.
\newblock {\em IEEE Transactions on Information Theory}, 51(5):1900 -- 1907,
  May 2005.

\bibitem{Holmes04}
R.~B. Holmes and V.~I. Paulsen.
\newblock Optimal frames for erasures.
\newblock {\em Linear Algebra Appl.}, 377:31--51, 2004.

\bibitem{Tropp05}
J.~A. Tropp.
\newblock Complex equiangular tight frames.
\newblock In {\em Proc. SPIE Wavelets XI}, volume 5914, pages 1 -- 11, 2005.

\bibitem{Sustik07}
M.~A. Sustik, J.~A. Tropp, I.~S. Dhillon, and R.~W. Heath, Jr.
\newblock On the existence of equiangular tight frames.
\newblock {\em Linear Algebra Appl.}, 426(2-3):619--635, 2007.

\bibitem{Bodmann09}
B.~G. Bodmann, V.~I. Paulsen, and M.~Tomforde.
\newblock Equiangular tight frames from complex {S}eidel matrices containing
  cube roots of unity.
\newblock {\em Linear Algebra Appl.}, 430(1):396--417, 2009.

\bibitem{Bodmann2010}
B.~G. Bodmann and H.~J. Elwood.
\newblock Complex equiangular {P}arseval frames and {S}eidel matrices
  containing {$p$}th roots of unity.
\newblock {\em Proc. Amer. Math. Soc.}, 138(12):4387--4404, 2010.

\bibitem{Waldron09}
S.~Waldron.
\newblock On the construction of equiangular frames from graphs.
\newblock {\em Linear Algebra Appl.}, 431(11):2228--2242, 2009.

\bibitem{Fickus2012}
M.~Fickus, D.~Mixon, and J.~Tremain.
\newblock {Steiner equiangular tight frames}.
\newblock {\em Linear Algebra Appl.}, 436(5):1014--1027, March 2012.

\bibitem{redmond2009existence}
D.~J. Redmond.
\newblock {\em Existence and construction of real-valued equiangular tight
  frames}.
\newblock PhD thesis, University of Missouri--Columbia, 2009.

\bibitem{Lemmens73}
P.~W.~H. Lemmens and J.~J. Seidel.
\newblock Equiangular lines.
\newblock {\em J. Algebra}, 24:494--512, 1973.

\bibitem{DGS75}
P.~Delsarte, J.~M. Goethals, and J.~J. Seidel.
\newblock Bounds for systems of lines and {J}acobi polynomials.
\newblock {\em Philips Res. Repts.}, 30(3):91--105, 1975.
\newblock Issue in honour of C.J. Bouwkamp.

\bibitem{Dattaa}
S.~Datta and J.~Oldroyd.
\newblock Construction of $k$-angle tight frames.
\newblock {\em Numerical Functional Analysis and Optimization}, 37(8):975 --
  989, 2016.

\bibitem{casazza2017toward}
P.~G. Casazza, A.~Farzannia, J.~I. Haas, and T.~T. Tran.
\newblock Toward the classification of biangular harmonic frames.
\newblock {\em Applied and Computational Harmonic Analysis}, 2017.

\bibitem{haas2017constructions}
J.~I. Haas, J.~Cahill, J.~Tremain, and P.~G Casazza.
\newblock Constructions of biangular tight frames and their relationships with
  equiangular tight frames.
\newblock {\em arXiv preprint arXiv:1703.01786}, 2017.

\bibitem{Barg2014}
A.~Barg, A.~Glazyrin, K.~A Okoudjou, and W.-H. Yu.
\newblock Finite two-distance tight frames.
\newblock {\em Linear Algebra Appl.}, 475:163--175, 2015.

\bibitem{Larman77}
D.~G. Larman, C.~A. Rogers, and J.~J. Seidel.
\newblock On two-distance sets in {E}uclidean space.
\newblock {\em Bull. London Math. Soc.}, 9(3):261--267, 1977.

\bibitem{Goyal_2001}
V.~K. Goyal, J.~Kovacevic, J.~A. Kelner, Communicated Henrique, and S.~Malvar.
\newblock Quantized frame expansions with erasures.
\newblock {\em Applied and Computational Harmonic Analysis}, 10:203 -- 233,
  2001.

\bibitem{Chr03}
O.~Christensen.
\newblock {\em An Introduction to Frames and {R}iesz Bases}.
\newblock Birkh\"{a}user, Boston, 2003.

\bibitem{Dau92}
I.~Daubechies.
\newblock {\em Ten Lectures on Wavelets}.
\newblock SIAM, Philadelphia, 1992.

\bibitem{BF03}
J.~J. Benedetto and M.~Fickus.
\newblock Finite normalized tight frames.
\newblock {\em Advances in Computational Mathematics}, 18:357--385, 2003.

\bibitem{bodmann2015achieving}
B.~Bodmann and J.~Haas.
\newblock Achieving the orthoplex bound and constructing weighted complex
  projective 2-designs with {S}inger sets.
\newblock {\em arXiv preprint arXiv:1509.05333}, 2015.

\bibitem{moore2013difference}
E.~H. Moore and H.~S. Pollatsek.
\newblock {\em Difference Sets: Connecting Algebra, Combinatorics, and
  Geometry}, volume~67.
\newblock American Mathematical Soc., 2013.

\bibitem{colbourn2010crc}
C.~J. Colbourn.
\newblock {\em CRC handbook of combinatorial designs}.
\newblock CRC press, 2010.

\bibitem{bose1949}
R.~C. Bose.
\newblock A note on {F}isher's inequality for balanced incomplete block
  designs.
\newblock {\em Ann. Math. Statist.}, 20(4):619--620, 12 1949.

\bibitem{majumdar1953some}
K.~N. Majumdar.
\newblock On some theorems in combinatorics relating to incomplete block
  designs.
\newblock {\em The Annals of Mathematical Statistics}, 24(3):377--389, 1953.

\bibitem{Beutelspacher82}
A.~Beutelspacher.
\newblock On extremal intersection numbers of a block design.
\newblock {\em Discrete Mathematics}, 42:37--49, 1982.

\bibitem{cochran1957experimental}
W.~G. Cochran and G.~M. Cox.
\newblock {\em Experimental Designs}.
\newblock Wiley Publications in Statistics. John Wiley \& Sons, 2 edition,
  1957.

\end{thebibliography}

\end{document}